\documentclass[11pt,oneside,english]{amsart}
\usepackage[T1]{fontenc}
\usepackage[latin9]{inputenc}
\usepackage[a4paper]{geometry}
\usepackage{babel}
\usepackage{verbatim}
\usepackage{url}
\usepackage{amsthm}
\usepackage{graphicx}
\usepackage{cite}
\usepackage[unicode=true,pdfusetitle,
 bookmarks=true,bookmarksnumbered=false,bookmarksopen=false,
 breaklinks=false,pdfborder={0 0 0},backref=false,colorlinks=false]
 {hyperref}
\usepackage{cite}

\makeatletter

\providecommand{\tabularnewline}{\\}

\numberwithin{equation}{section}
\numberwithin{figure}{section}
  \theoremstyle{definition}
  \newtheorem{defn}{\protect\definitionname}
  \theoremstyle{plain}
  \newtheorem*{thm*}{\protect\theoremname}

\makeatother

  \providecommand{\definitionname}{Definition}
  \providecommand{\theoremname}{Theorem}

\begin{document}

\title{On Fair Size-Based Scheduling}

\author{Matteo Dell'Amico, Damiano Carra, and Pietro Michiardi}

\begin{abstract}
By executing jobs serially rather than in parallel, size-based
scheduling policies can shorten time needed to complete jobs; however,
major obstacles to their applicability are fairness guarantees and the
fact that job sizes are rarely known exactly a-priori. Here, we
introduce the $\mathrm{Pri}$ family of size-based scheduling policies;
$\mathrm{Pri}$ simulates \emph{any} reference scheduler and executes
jobs in the order of their simulated completion: we show that these
schedulers give strong fairness guarantees, since \emph{no job
  completes later in $\mathrm{Pri}$ than in the reference policy}. In
addition, we introduce PSBS, a practical implementation of such a
scheduler: it works online (i.e., without needing knowledge of jobs submitted
in the future), it has an efficient $O(\log n)$ implementation and it
allows setting priorities to jobs. Most importantly, unlike earlier
size-based policies, the performance of PSBS degrades gracefully with
errors, leading to performances that are close to optimal in a variety
of realistic use cases.
\end{abstract}

\thanks{Matteo Dell'Amico is with Symantec Research Labs; this work
  was done while he was at EURECOM. Email:
  matteo\_dellamico@symantec.com. Phone: +33~4~93~00~82~61. Address:
  Symantec Research Labs at EURECOM, Campus SophiaTech, 450 Route des
  Chappes, 06410 Biot, France.}

\thanks{Damiano Carra is with University of Verona.  Email:
  damiano.carra@univr.it. Phone: +39~045~802~7059. Address: Strada Le
  Grazie 15, Verona, Italy.}

\thanks{Pietro Michiardi is with EURECOM. Email:
  pietro.michiardi@eurecom.fr. Phone: +33~4~93~00~81~45. Address:
  Campus SophiaTech, 450 Route des Chappes, 06410 Biot, France.}

\maketitle

\newpage

\section{Introduction}

Schedulers are often based on \emph{fair sharing}, 
where the resources are divided
among jobs according to some fairness concept. The simplest case is
processor sharing (PS), which partitions the resources equally 
among pending jobs at every instant.
However, if users care about job completion time, 
rather than instantaneous job progression,
sharing is not optimal: this is shown in FSP \cite{
Friedman2003}, a scheduler that optimizes job completion times
while providing strong fairness guarantees. FSP dominates PS,
i.e., \emph{no job will complete later in FSP than in PS}, and it is 
based on a simple idea: schedule the job that would complete first 
in PS. Here, we discuss two issues arising from implementing a 
scheduler inspired by FSP in a practical context \cite{Pastorelli2013
}.  

First, what if our concept of fairness is more elaborate than simple
equal sharing? Many real-world schedulers have a flexibility which
goes even beyond that of priority classes: e.g., the Hadoop
capacity scheduler \cite{Zaharia2009} applies a hierarchical concept
of fair sharing to guarantee resources to units within an
organization. We thus introduce a generalization of FSP's dominance
result: given \emph{any} scheduler, simulating it and executing jobs
one at a time according to the order in which they would complete
dominates the scheduler itself. Therefore, \emph{any} scheduler can
be used as a reference for fairness, and executing jobs serially is
always beneficial.

Second, what happens when job size is only known approximately?
Indeed, in practical settings, job sizes are rarely known a-priori. 
We thus introduce our work on scheduling based on inexact sizes and 
show that, if a scheduler has been designed without considering that
the information about job size may be inaccurate, estimation errors 
may have dramatic impact on the performance for different 
workload characteristics. On the other hand, if
the consequences of the estimation errors are  
properly addressed -- as we do in our proposal,
PSBS \cite{DellAmico2014} -- the scheduler performs close to 
optimally in a variety of workloads. PSBS is efficient
(its complexity is $O\left(\log n\right)$ compared to $O\left(n\right)$
of FSP), and allows 
job priorities.

We conclude by highlighting open questions and future research directions 
related to scheduling with inexact job sizes.

\section{\label{sec:Analytical-Results}Dominance Results With Known Job Sizes}

We consider here the single-machine scheduling problem with release
times and preemption; our goal, that materializes in the $\mathrm{Pri
}$ scheduler, is to minimize the sum of 
completion times (according to Graham et al. \cite{
Graham1979}, the $1|r_{i};pmtn|\sum C_{i}$ problem) with the 
additional dominance requirement: no job should complete later than 
in a scheduler which is taken as a reference for fairness. Without 
this limitation, the optimal solution is the Shortest Remaining 
Processing Time (SRPT) policy. We call \emph{
schedule }a function $\omega\left(i,t\right)$
that outputs the fraction of system resources allocated to job $i$
at time $t$. For example, for the processor-sharing (PS) scheduler,
when $n$ jobs are \emph{pending }(released and not yet completed),
$\omega\left(i,t\right)=\frac{1}{n}$ if job $i$ is pending and $0$
otherwise. Furthermore, we call $C_{i,\omega}$ the completion time
of job $i$ under schedule $\omega$.
\begin{defn}
Schedule $\omega$ \emph{dominates} schedule $\omega'$ if $C_{i,\omega}\leq C_{i,\omega'}$
for each job $i$.
\end{defn}
Our scheduler prioritizes jobs according to the order in which they
complete in $\omega$.
\begin{defn}
A \emph{completion sequence $S=\left[s_{1},\ldots,s_{n}\right]$ }is
an ordering of the jobs to be scheduled. A schedule $\omega$ \emph{has
completion sequence} $S$ if $C_{s_{i},\omega}\leq C_{s_{j},\omega}\forall i<j$.
\end{defn}

\begin{defn}
For a completion sequence $S$, the $\mathrm{Pri}_{S}$ schedule is
such that $\mathrm{Pri}_{S}\left(i,t\right)=1$ if $i$ is the first
pending job to appear in $S$; $\mathrm{Pri}_{S}\left(i,t\right)=0$
otherwise.
\end{defn}
We now show that scheduling jobs in the order in which they complete
under $\omega$ dominates $\omega$.
\begin{thm*}
$\mathrm{Pri}_{S}$ dominates any schedule with completion sequence
$S$.\end{thm*}
\begin{proof}
We have to show that $C_{i,\mathrm{Pri}_{S}}\leq C_{i,\omega}$ for
each job $i$ and any schedule $\omega$ with completion sequence
$S$. Let $j$ be the position of $j$ in $S$ (i.e., $i=s_{j}$); we call $M$ the minimal makespan of the $S_{\leq j}=\left\{ s_{1},\ldots,s_{j}\right\} $
set of jobs,%
\footnote{The \emph{makespan} of a set of jobs is the maximum among their completion
times, therefore $M=\min_{\omega\in\Omega}\max_{i\in\left\{ 1,\ldots,j\right\} }C_{S_{i},\omega}$
where $\Omega$ is the set of all possible schedules.%
} and we show that $C_{i,\mathrm{Pri}_{S}}\leq M$ and $M\leq C_{i,\omega}$:
\begin{itemize}
\item $C_{i,\mathrm{Pri}_{S}}\leq M$: minimizing the makespan of $S_{\leq j}$
is equivalent to solving the $1|r_{i};pmtn|C_{\max}$ problem applied
to the jobs in $S_{\leq j}$: this is guaranteed if all resources
are assigned to jobs in $S_{\leq j}$ as long as any of them are pending
\cite{Liu1973}. $\mathrm{Pri}_{S}$ guarantees this, hence the makespan
of $S_{\leq j}$ using $\mathrm{Pri}_{S}$ is $M$. Since $i\in S_{\leq j}$,
$C_{i,\mathrm{Pri}_{S}}\leq M$.
\item $M\leq C_{i,\omega}$ follows trivially from the fact that $\omega$
has completion sequence $S$ and, therefore, $C_{i,\omega}$ is the
makespan for $S_{\leq j}$ using schedule $\omega$.
\qedhere
\end{itemize}
\end{proof}
This theorem generalizes the results by Friedman and Henderson
 \cite{Friedman2003}: FSP follows from applying $\mathrm{Pri}_{S}$
to the completion sequence of PS.
The generalization is important: in practice, one can define 
a scheduler that provides a desired type of fairness, and then optimize
the performance in terms of completion time by applying the $\mathrm{Pri}_{S}$
scheduler. For instance, assume that the system deals with different 
classes of jobs that have different weights, and the scheduler to apply 
to provide fairness is the discriminatory processor sharing (DPS): the 
theorem guarantees that $\mathrm{Pri}_{S}$ dominates DPS. We have exploited
exactly this results in our scheduler PSBS \cite{DellAmico2014}, which, 
in the absence of errors, dominates DPS.
Note that both FSP and
PSBS are appliable \emph{online}: even without information on
future jobs, it is possible to compute which pending job completes
first in PS and DPS and hence decide which job to schedule.

\section{\label{sec:Scheduling-With-Uncertain}Scheduling With Approximated
Sizes}

The results above seem to suggest that size-based schedulers should
be employed ubiquitously; however, a major obstacle to their applicability
is that, in a large majority of cases, job size cannot be known exactly;
on the other hand, it is often possible to compute an \emph{estimate}. 
In this Section, we synthetize our results
on the topic \cite{dell2014revisiting,DellAmico2014}.%

Only a few other works \cite{lu2004size,wierman2008scheduling} tackle
the problem of scheduling with inaccurate sizes; they show rather pessimistic
results, suggesting that size-based schedulers outperform non
size-based counterparts only when estimations are precise. We have
complemented those works with an extensive simulative study, generating
synthetic workloads with several varying parameters related to the job
size distribution and the error distribution. To help reproducibility,
our simulator is available as free software.%
\footnote{\url{http://github.com/bigfootproject/schedsim}}
We have found that job size distribution plays an essential role:
if job sizes are skewed (i.e., a few very large jobs make up a large
fraction of the total work), size estimation errors cause serious
performance issues in existing size-based schedulers. This phenomenon
is mainly caused by the fact that, if the size of a large job is 
underestimated, it gains positions in the queue; when it enters in 
the service, it reaches a point when it cannot be preempted and it blocks the server
until it has completed (at the expense of jobs that are actually small).
The opposite situation, a job with overestimated size, instead, has little
impact on the other jobs (see \cite{DellAmico2014} for an illustrative example).

Our proposal, PSBS, leverages these last observations. The main idea is
to react when the system detects that a job has been underestimated, 
i.e., when a job is taking more resources than the ones initially estimated --
we call these jobs \emph{late}. The scheduler treats the late jobs differently, 
and lets other jobs be served, so that the impact of the underestimation 
is limited.

This, coupled with the fact that PSBS schedules jobs one at a time 
according to their completion time when simulating DPS, 
allows PSBS to obtain close to optimal performance, yet guaranteeing 
the same fairness of the simulated scheduling (DPS). 
PSBS has been inspired by FSP but, unlike FSP, it allows setting priorities
and it deals with estimation errors. Moreover, PSBS
is efficient, since its complexity is $O\left(\log n\right)$, compared to
$O\left(n\right)$ of FSP. Note that, by properly configuring its parameters and 
with no estimation errors, PSBS behaves as FSP, therefore our efficient 
implementation represents a gain with respect to FSP.

\begin{figure}
\begin{centering}
\begin{tabular}{ccc}
\includegraphics[width=0.3\textwidth]{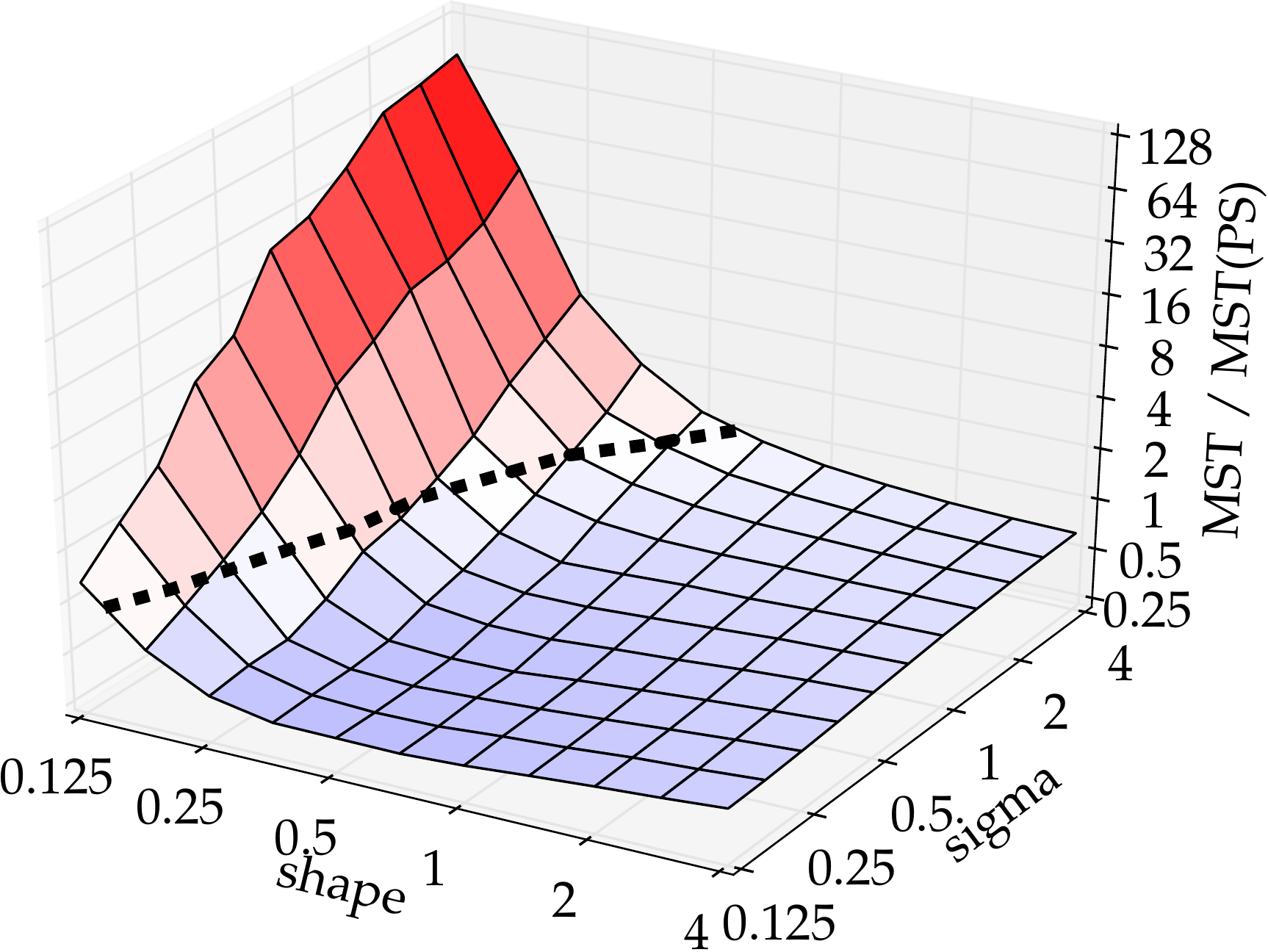} & \includegraphics[width=0.3\textwidth]{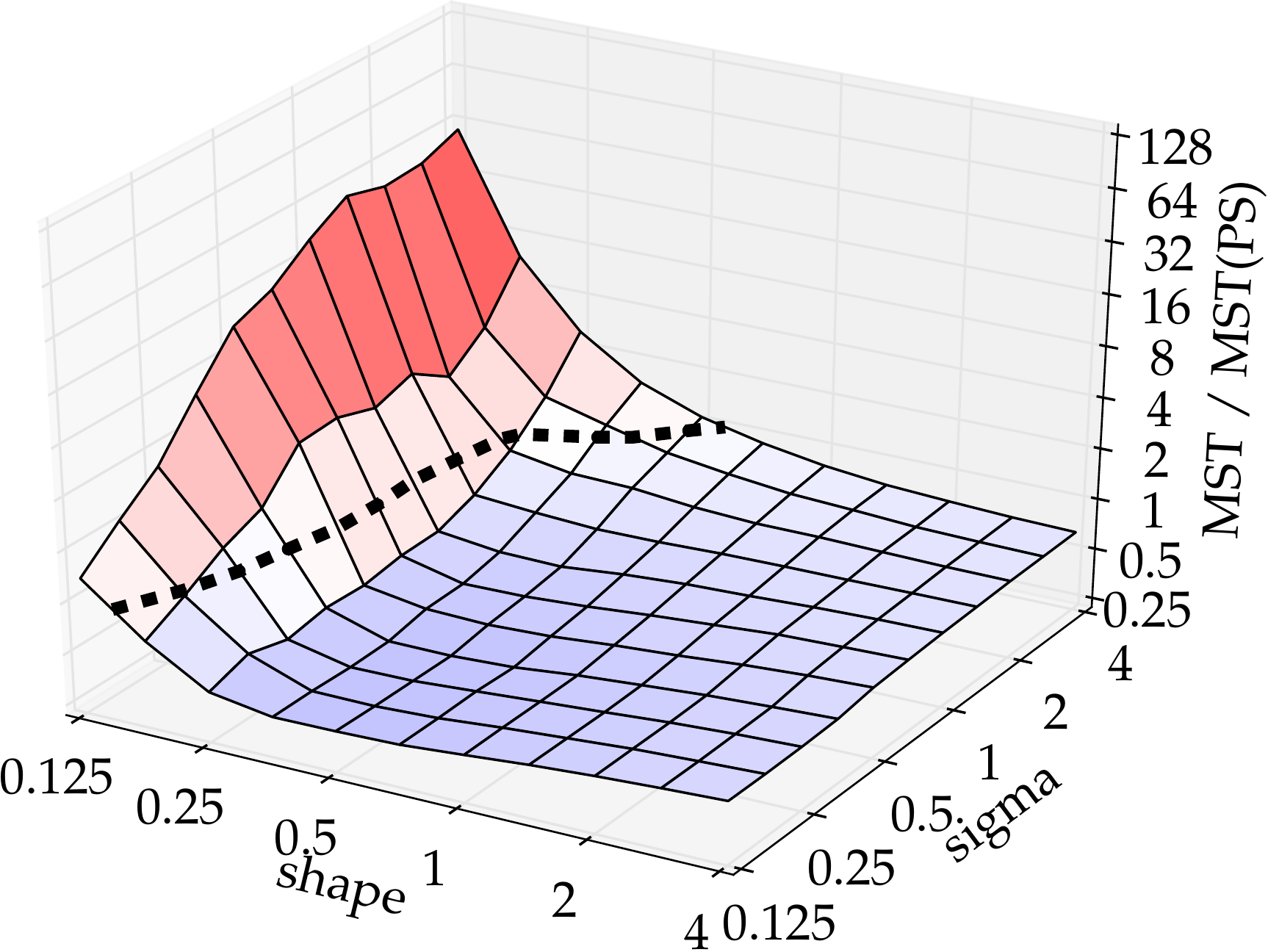} & \includegraphics[width=0.3\textwidth]{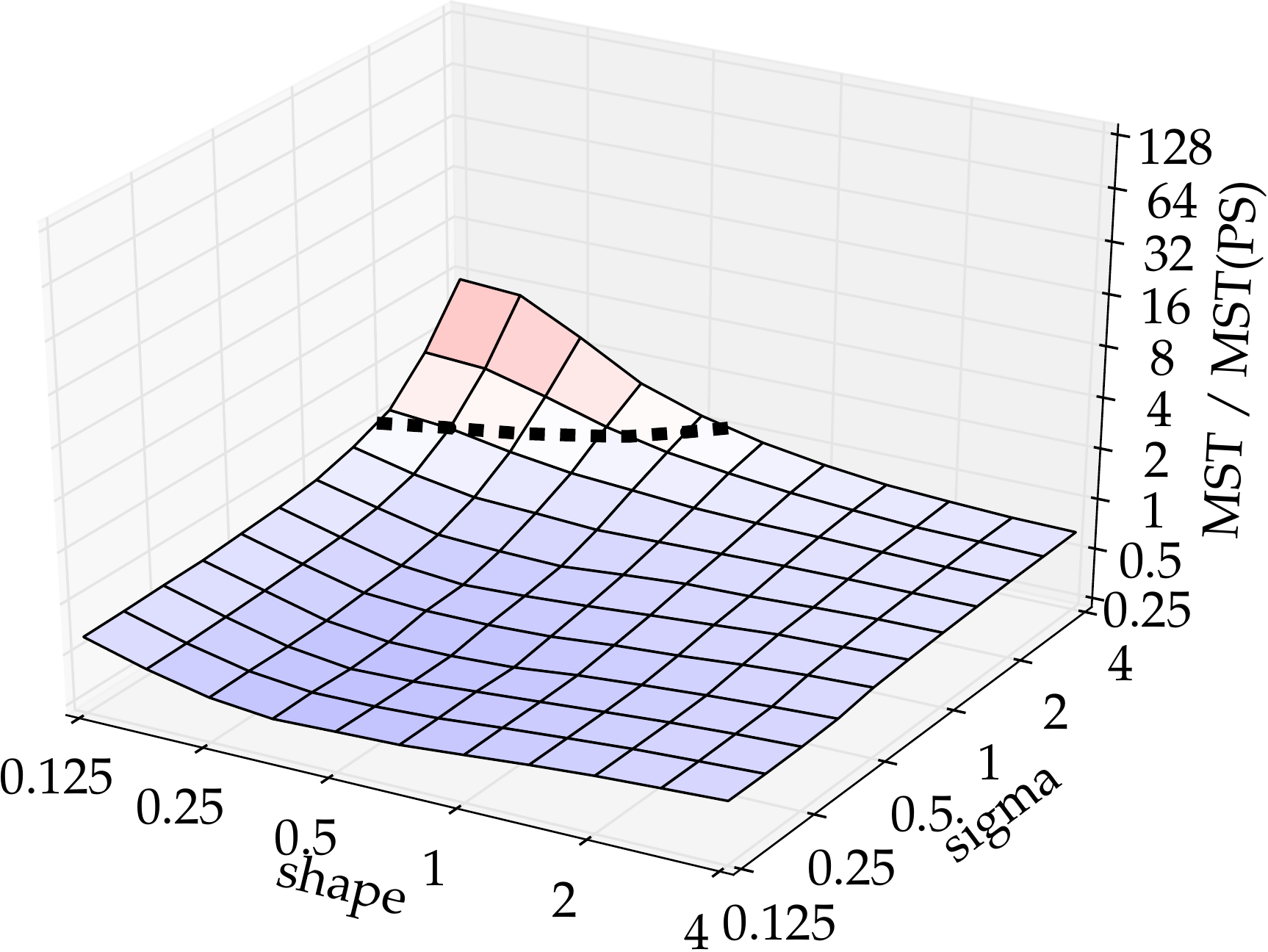}\tabularnewline
SRPT & FSP & PSBS\tabularnewline
\end{tabular}
\par\end{centering}

\caption{\label{fig:Mean-sojourn-time}Mean sojourn time using size-based schedulers,
normalized against PS.}
\end{figure}

In Figure \ref{fig:Mean-sojourn-time}, we show the ratio between
the mean sojourn time (MST)%
\footnote{A job's \emph{sojourn time} is the interval between its release
and completion; minimizing MST also minimizes $\sum C_{i}$.%
} obtained using size-based schedulers, such as SRPT and FSP, 
normalized against the MST
of PS. Job priorities are homogeneous and sizes are generated according
to a Weibull distribution (heavy-tailed for shape < 1, light-tailed
otherwise); the relative job size estimation error is distributed
according to a log-normal distribution -- i.e., higher values of sigma
yield larger errors. Job inter-arrival times are distributed exponentially
and the load (ratio between arrival and service rate) is set
to 0.9.

The results of Figure \ref{fig:Mean-sojourn-time} show that SRPT and FSP
suffer when the workload is highly skewed, even with moderate estimation errors; 
on the other hand, PSBS largely corrects this issue,
and it is outperformed by PS only in extreme cases 
where the workload is very skewed \emph{and} job size estimation is 
very imprecise (sigma greater than 2, which corresponds to a correlation 
coefficient between the job size and its estimate less than 0.15). 
In fact, PSBS performs close to optimally in most cases; similar results
are obtained when playing back real workloads on the simulator. Other
parameters such as load and job inter-arrival times do not have a
large impact on the results.

\section{\label{sec:Conclusion}Conclusion}

We have shown two main results. For the first time in this paper,
we have generalized the dominance results of FSP over PS to \emph{any} ``simulated'' scheduler. We also have synthetized our results on scheduling based on approximate job sizes: PSBS largely mitigates the problem of existing policies, and often obtains close to optimal results, yet maintaining the desired fairness among jobs.
We consider that these results can be interesting both for theorists
and for practitioners. In practical systems, the PSBS policy can be used as a basis to build efficient size-based schedulers, as our work for Hadoop \cite{Pastorelli2013} demonstrates.



With respect to theory, we identify a problem that may be of 
interest to the community. We have shown that a scheduler 
such as PSBS can perform well in a variety of realistic use cases; 
we are able to explain these results with intuition and substantiate 
them with numerical experiments. However, 
an analytical characterization of the problem may provide more
insights on the situations where such a strategy performs 
well, and a way to predict scheduler performance in function of the 
workload characteristics.
Such a modeling approach could be useful to go beyond PSBS.
We speculate that job size information -- even if
approximated -- is 
better than no information, and hence we
conjecture that it is possible to design a scheduler that \emph{always} 
outperforms non size-based counterparts, when the distribution of estimation
errors is known.

\bibliographystyle{abbrv}
\bibliography{references}

\begin{thebibliography}{1}

\bibitem{DellAmico2014}
M.~Dell'Amico et~al.
\newblock {PSBS}: Practical size-based scheduling.
\newblock {\em arXiv:1410.6122}, 2014.

\bibitem{dell2014revisiting}
M.~Dell'Amico et~al.
\newblock Revisiting size-based scheduling with estimated job sizes.
\newblock In {\em MASCOTS}. IEEE, 2014.

\bibitem{Friedman2003}
E.~J. Friedman and S.~G. Henderson.
\newblock Fairness and efficiency in web server protocols.
\newblock In {\em ACM SIGMETRICS Performance Evaluation Review}, volume~31,
  pages 229--237. ACM, 2003.

\bibitem{Graham1979}
R.~L. Graham et~al.
\newblock Optimization and approximation in deterministic sequencing and
  scheduling: a survey.
\newblock {\em Annals of discrete Mathematics}, 5:287--326, 1979.

\bibitem{Liu1973}
C.~L. Liu and J.~W. Layland.
\newblock Scheduling algorithms for multiprogramming in a hard-real-time
  environment.
\newblock {\em JACM}, 20(1):46--61, 1973.

\bibitem{lu2004size}
D.~Lu et~al.
\newblock Size-based scheduling policies with inaccurate scheduling
  information.
\newblock In {\em MASCOTS}. IEEE, 2004.

\bibitem{Pastorelli2013}
M.~Pastorelli et~al.
\newblock {HFSP}: size-based scheduling for {H}adoop.
\newblock In {\em BIGDATA}. IEEE, 2013.

\bibitem{wierman2008scheduling}
A.~Wierman and M.~Nuyens.
\newblock Scheduling despite inexact job-size information.
\newblock In {\em ACM SIGMETRICS Performance Evaluation Review}, volume~36,
  pages 25--36. ACM, 2008.

\bibitem{Zaharia2009}
M.~Zaharia.
\newblock Job scheduling with the fair and capacity schedulers.
\newblock In {\em Hadoop Summit}, 2009.

\end{thebibliography}

\end{document}